\documentclass{article}
\usepackage{fullpage}

\usepackage{algorithm,algorithmicx}
\usepackage{amssymb}
\usepackage[noend]{algpseudocode}

\usepackage{tikz}
\usepackage{amsthm}
\usepackage{thmtools}
\usepackage{thm-restate}

\usepackage{hyperref}
\newtheorem{theorem}{Theorem}
\newtheorem{lemma}{Lemma}

\usepackage{etoolbox}  
\makeatletter
\patchcmd{\algorithmic}{\addtolength{\ALC@tlm}{\leftmargin} }{\addtolength{\ALC@tlm}{\leftmargin}}{}{}
\makeatother

\usepackage{graphicx} 

\usepackage{xcolor}

\usepackage{amsmath}
\usepackage{cleveref}

\usepackage{nicefrac}

\usepackage{algorithm}

\begin{document}
\title{Interval Selection in Sliding Windows}

\date{}

\author{Cezar-Mihail Alexandru\footnote{Supported by EPSRC Doctoral Training Studentship EP/T517872/1.} \\ School of Computer Science, University of Bristol, Bristol, UK \\ ca17021@bristol.ac.uk\\
\and Christian Konrad\footnote{Supported by EPSRC New Investigator Award EP/V010611/1.} \\
School of Computer Science, University of Bristol, Bristol, UK \\
christian.konrad@bristol.ac.uk}

\newcommand{\Exp}{{\mathbb{E}}}
\newcommand{\Expmath}{\mathop{\mathbb{E}}}
\newcommand{\Order}{\mathrm{O}}
\newcommand{\OrderT}{\tilde{\mathrm{O}}}
\newcommand{\OmegaT}{\tilde{\mathrm{\Omega}}}
\newcommand{\ThetaT}{\tilde{\Theta}}
\newcommand{\polylog}{\mathop{\mathrm{polylog}}\nolimits}
\newcommand{\poly}{\mathop{\mathrm{poly}}\nolimits}

\baselineskip=1.1\baselineskip

\maketitle
\allowdisplaybreaks

\thispagestyle{empty}

\begin{abstract}
     We initiate the study of the \textsf{Interval Selection} problem in the (streaming) sliding window model of computation. 
     In this problem, an algorithm receives a potentially infinite stream of intervals on the line, and the objective is to maintain at every moment an approximation to a largest possible subset of disjoint intervals among the $L$ most recent intervals, for some integer $L$.

     We give the following results:
     \begin{enumerate}
         \item In the unit-length intervals case, we give a $2$-approximation sliding window algorithm with space $\OrderT(|OPT|)$, and we show that any sliding window algorithm that computes a $(2-\varepsilon)$-approximation requires space $\Omega(L)$, for any $\varepsilon > 0$.
         \item In the arbitrary-length case, we give a $(\frac{11}{3}+\varepsilon)$-approximation sliding window algorithm with space $\OrderT(|OPT|)$, for any constant $\varepsilon > 0$, which constitutes our main result.\footnote{We use the notation $\tilde{O}(.)$ to mean  $O(.)$ where $\polylog$ factors and dependencies on $\varepsilon$ are suppressed.} 
         We also show that space $\Omega(L)$ is needed for algorithms that compute a $(2.5-\varepsilon)$-approximation, for any $\varepsilon > 0$.
     \end{enumerate}

     Our main technical contribution is an improvement over the smooth histogram technique, which consists of running independent copies of a traditional streaming algorithm with different start times. By employing the one-pass $2$-approximation streaming algorithm by Cabello and P\'{e}rez-Lantero [Theor. Comput. Sci. '17] for \textsf{Interval Selection}  on arbitrary-length intervals as the underlying algorithm, the smooth histogram technique immediately yields a $(4+\varepsilon)$-approximation in this setting. Our improvement is obtained by forwarding the structure of the intervals identified in a run to the subsequent run, which constrains the {\em shape} of an optimal solution and allows us to target optimal intervals differently. 
\end{abstract}

\newpage

\pagenumbering{arabic}

\section{Introduction}
\subparagraph*{Sliding Window Model} The {\em sliding window model} of computation introduced by Datar et al. \cite{dgim02} captures many of the challenges that arise when processing infinite data streams. In this model, an algorithm receives an infinite stream of data items and is required to maintain, at every moment, a solution to a given problem on the current {\em sliding window}, i.e., on the $L$ most recent data items, for an integer $L$. The objective is to design algorithms that use much less space than the size of the sliding window $L$. 

Many modern data sources are best modelled as infinite data streams rather than as data sets of large but finite sizes. For example, the sequence of Tweets on X (formerly Twitter), the sequence of IP packages passing through a network router, and continuous sensor measurements for monitoring the physical world are a priori unending. Such data sets typically constitute time-series data, where the resulting data stream is ordered with respect to the data items' creation times. When processing such streams, it is reasonable to focus on the most recent data items (as it is modelled in the sliding window model by the sliding window size $L$) since the near past usually affects the present more strongly than older data.

The sliding window model should be contrasted with the more traditional  {\em one-pass data streaming model}. In the data streaming model, an algorithm processes a finite stream of $n$ data items and is tasked with producing a single output once all items have been processed. Similar to the sliding window model, the objective is to design algorithms that use as little space as possible, in particular, sublinear in the length of the stream. Since sliding window algorithms with $L = n$ can immediately be used in the data streaming model, problems are generally harder to solve in the sliding window model.

\subparagraph*{\textsf{Interval Selection} Problem}
In this work, we initiate the study of the \textsf{Interval Selection} problem in the sliding window model. Given a set $\mathcal{S}$ of $n$ intervals on the real line, the objective is to find a subset $\mathcal{I} \subseteq \mathcal{S}$ of pairwise non-overlapping intervals of maximum cardinality. The problem can also be regarded as the \textsf{Maximum Independent Set} problem in the interval graph associated with the intervals $\mathcal{S}$. We consider both the unit-length case, where all intervals are of length $1$, and the arbitrary-length case, where no restriction on the lengths of the intervals is imposed. 

\textsf{Interval Selection} is fully understood in the one-pass streaming model. Emek et al. \cite{ehr16} gave a $\frac{3}{2}$-approximation streaming algorithm for unit-length intervals and a $2$-approximation streaming algorithm for arbitrary-length intervals. Both algorithms use space $O(|OPT|)$, where $OPT$ denotes an optimal solution, assuming that the space required for storing an interval is $O(1)$. Emek et al. also gave matching lower bounds, showing that, for both the unit-length and the arbitrary-length case, slightly better approximations require space $\Omega(n)$. Subsequently, Cabello and P\'{e}rez-Lantero \cite{cp17} also gave algorithms for the unit-length and the arbitrary-length cases that match the guarantees of those by Emek et al. but are significantly simpler. We will reuse one of the algorithms by Cabello and P\'{e}rez-Lantero in this paper. Last, weighted intervals as well as the insertion-deletion setting, where previously inserted intervals can be deleted again, have also  been considered \cite{ddk23, Bakshi2019WeightedMI}, where \cite{Bakshi2019WeightedMI} addresses the challenge of outputting the size or weight of a largest/heaviest independent set rather than outputting the intervals themselves.

\subparagraph*{The Smooth Histogram Technique} 
Braverman and Ostrovsky \cite{bO07} introduced the smooth histogram technique, which allows deriving sliding window algorithms from traditional streaming algorithms at the expense of slightly increased space requirements and approximation guarantees. The method works as follows. Given a streaming algorithm $\mathcal{A}$ for a specific problem $\textsf{P}$ that fulfills certain smoothness properties (see \cite{bO07} for details), multiple copies of $\mathcal{A}$ are run with different starting positions in the stream. The runs are such that consecutive runs differ only slightly in solution quality, and, thus, when a run expires due to the fact that its starting position fell out of the current sliding window, the subsequent run can be used to still yield an acceptable solution. The smooth histogram technique can be applied to the \textsf{Interval Selection} algorithms by Emek et al. \cite{ehr16} and by Cabello and P\'{e}rez-Lantero \cite{cp17}, and we immediately obtain sliding window algorithms for both the unit-length and the arbitrary-length cases using space $\tilde{O}(|OPT|)$\footnote{We use the notation $\tilde{O}(.)$ to mean  $O(.)$ where $\polylog$ factors and dependencies on $\varepsilon$ are suppressed.}. For unit-length intervals, the resulting approximation factor is $3+\varepsilon$, for any $\varepsilon > 0$, and for arbitrary-length intervals, the approximation factor is $4+\varepsilon$, for any $\varepsilon > 0$. We will provide the analysis of the $(4+\varepsilon)$-approximation for arbitrary-length intervals in this paper (\textbf{Theorem~\ref{thm:smooth-histogram}}) since it forms the basis of the analysis of one of our algorithms. 

\subparagraph*{Our Results} In this work, we show that it is possible to improve upon the guarantees obtained from the smooth histogram technique.
We give deterministic sliding window algorithms and lower bounds that also apply to randomized algorithms for \textsf{Interval Selection} for both the unit-length and arbitrary-length cases. Our algorithms use space $\tilde{O}(|OPT|)$ at any moment during the processing of the stream, where $OPT$ denotes an optimal solution in the current sliding window. Observe that $OPT$ may vary throughout the processing of the stream, and, thus, the space used by our algorithms may therefore also change accordingly.

Regarding unit-length intervals, we give a $2$-approximation sliding window algorithm using $O(|OPT|)$ space, and we prove that any sliding window algorithm with an approximation guarantee of $2-\varepsilon$, for any $\varepsilon > 0$, requires space $\Omega(L)$. Recall that, in the streaming model, a $\frac{3}{2}$-approximation can be achieved with space $O(|OPT|)$. Our lower bound thus establishes a separation between the sliding window and the streaming models for unit-length intervals.

In the arbitrary-length case, we give a $(\frac{11}{3}+\varepsilon)$-approximation sliding window algorithm with space $\OrderT(|OPT|)$, improving over the smooth histogram technique, which constitutes our main and most technical result. We also prove that any $(\frac{5}{2} - \varepsilon)$-approximation algorithm, for any $\varepsilon > 0$, requires space $\Omega(L)$. Since, in the streaming model, a $2$-approximation can be achieved with space $O(|OPT|)$, our lower bound also establishes a separation between the sliding window and the streaming models in the arbitrary-length case. 

    We summarize and contrast our results with results from the streaming model in Figure~\ref{fig:results}.

\begin{figure}[ht]
    \begin{tabular}{l|cc|cc}
     & \multicolumn{2}{c|}{\textbf{Streaming model} \cite{ehr16,cp17}}  & \multicolumn{2}{c}{\textbf{Sliding window model} (this paper)}  \\
    & Algorithm  & LB & Algorithm & LB \\
    \hline 
 Unit-length Intervals & $\frac{3}{2}$ & $\frac{3}{2} - \varepsilon$ & $2$ (\textbf{Thm~\ref{thm:ub-unit-length}}) & $2 - \varepsilon$ (\textbf{Thm~\ref{thm:lb-unit-length}}) \\
 Arbitrary-length Intervals & $2$ & $2 - \varepsilon$ & $\frac{11}{3}$ (\textbf{Thm~\ref{thm:ub-arbitrary-length}}) & $\frac{5}{2} - \varepsilon$ (\textbf{Thm~\ref{thm:lb-arbitrary-length}})
 \end{tabular}
 \caption{Approximation factors achievable in the streaming and sliding window models. All algorithms use space $\OrderT(|OPT|)$, while all lower bound results are to be interpreted in that achieving the stated approximation guarantee requires space $\Omega(n)$ (streaming) or $\Omega(L)$ (sliding window model). The lower bound results hold for any $\varepsilon > 0$. \label{fig:results}}
\end{figure}

\subparagraph{A Lack of Lower Bounds in the Sliding Window Model} Interestingly, to the best of our knowledge, we are aware of only a single lower bound result for graphs in the sliding window model published before our work. Crouch et al. \cite{cms13} proved an $\Omega(L)$ space lower bound for \textsf{Minimum Spanning Forest} in the sliding window model, while this problem can be solved with space $\OrderT(n)$ in the streaming model.

Our work establishes such separation results for both the unweighted and the weighted versions of \textsf{Interval Selection}. While our results for arbitrary-length intervals are not tight, we stress that for most problems considered, including  \textsf{Maximum Matching} and \textsf{Minimum Vertex Cover}, no tight bounds are known. It is unclear whether this is due to a lack of techniques for improved algorithms or for stronger lower bounds.


\subparagraph{Techniques}
We will first discuss the key ideas behind our results for unit-length intervals, and then discuss our results for arbitrary-length intervals.

\textit{Unit-length Intervals.} Our algorithm for unit-length intervals is surprisingly simple yet optimal, as established by our lower bound result. For each integer $r$, maintain the latest interval within the current sliding window whose left endpoint lies in the interval $[r, r+1)$ if there is one. We argue that, if at any moment, the algorithm stores $D$ intervals, then we can extract an independent set of size at least $D/2$ by considering either only the intervals $[r, r+1)$ where $r$ is odd or where $r$ is even, while $OPT$ is bounded by $D/2 \le OPT \le D$, which establishes both the approximation factor of $2$ and the space requirements. We note that the idea of considering either only the odd or even intervals for obtaining a $2$-approximation was previously used by \cite{Bakshi2019WeightedMI}.

Our lower bound for unit-length intervals is obtained by a reduction to the $\textsf{Index}_n$ problem in the one-way two-party communication setting. 
In this setting, there are two parties, denoted Alice and Bob. 
Each party holds a portion of the input data. Alice sends a single message to Bob, who then outputs the result of the computation. The objective is to solve a problem using a message of smallest possible size. In $\textsf{Index}_n$, Alice holds a bit-string $X \in \{0, 1\}^n$, and Bob holds an index $J \in [n]$, where $[n] = \{1,2,3,...,n\}$, and the objective for Bob is to report the bit $X[J]$. It is well-known that a message of size $\Omega(n)$ is needed to solve the problem. 

We argue that a sliding window algorithm $\mathcal{A}$ for \textsf{Interval Selection} on unit-length intervals with approximation guarantee slightly below $2$ can be used to solve $\textsf{Index}_{\Theta(L)}$. To this end, Alice translates the bit-string $X$ into a {\em clique gadget}, i.e., a stack of overlapping $\Theta(L)$ interval slots that are slightly shifted from left-to-right, where interval $i$ is present in the stack if and only if $X[i] = 1$. Clique gadgets have been used in all previous space lower bound constructions for intervals \cite{ehr16,Bakshi2019WeightedMI,ddk23}. Alice then runs $\mathcal{A}$ on these intervals and sends the memory state of $\mathcal{A}$ to Bob. Bob subsequently feeds an interval located slightly to the right of the slot of interval $J$ into the execution of $\mathcal{A}$ such that Bob's interval overlaps with all interval slots at positions $\ge J+1$ and does not overlap with all interval slots at positions $\le J$. The key idea of this reduction is that, since $\mathcal{A}$ is a sliding window algorithm, it must be able to report a valid solution even if any prefix of intervals of the stack are deleted/have expired. Consider thus the situation when the intervals that are located in the first $J-1$ slots have expired. Then, the resulting instance has an independent set of size $2$ if and only if $X[J] = 1$, otherwise a largest independent set is of size $1$. Since the approximation factor of $\mathcal{A}$ is below $2$, $\mathcal{A}$ can thus distinguish between the two cases and solve $\textsf{Index}_{\Theta(L)}$. Since Alice only sent the memory state of $\mathcal{A}$ to Bob, we also obtain a space lower bound for $\mathcal{A}$. While this description covers the key idea of our lower bound, we note that our actual construction is slightly more involved due to an additional technical challenge. See proof of Theorem~\ref{thm:lb-unit-length} for details. 

Our lower bound construction shares similarities with the lower bounds by \cite{Bakshi2019WeightedMI} and \cite{ehr16}, as both of these lower bounds also work with clique gadgets and special intervals that render a specific interval in the clique gadget important. In \cite{Bakshi2019WeightedMI}, a reduction to the \textsf{Augmented-Index} problem is given in order to obtain a space lower bound for the {\em dynamic streaming setting}, where previously inserted intervals can be deleted again at any moment. In \textsf{Augmented-Index}, besides the index $J$, Bob also holds the prefix $X[1, \dots, J-1]$. While in our setting, intervals are deleted due to the shifting sliding window, in their lower bound, intervals are explicitly deleted by Bob.

\textit{Arbitrary-length Intervals.} Our algorithm and our lower bound for arbitrary-length intervals are substantially more involved, and our $(\frac{11}{3}+\varepsilon)$-approximation algorithm constitutes the main technical result of this paper. 

Our algorithm constitutes an improvement over the smooth histogram method. Using the one-pass $2$-approximation streaming algorithm for arbitrary-length intervals by Cabello and P\'{e}rez-Lantero \cite{cp17}, which we abbreviate by $\mathcal{CP}$, as the base algorithm of the smooth histogram method, we immediately obtain a $(4+\varepsilon)$-approximation sliding window algorithm using $\OrderT(|OPT|)$ space. 
The key idea of the method is to maintain various runs of $\mathcal{CP}$ with different starting times that are sufficiently spaced out so that only a logarithmic number of runs are needed, yet adjacent runs still have similar output sizes. Then, when a run expires, the subsequent run can still be used to report a good enough solution.

We observe that the executions of $\mathcal{CP}$ in the smooth histogram method are independent. Our key contribution that gives rise to our improvement is to forward the structure identified in a run of $\mathcal{CP}$ to the subsequent run. 
The $\mathcal{CP}$ algorithm, which we will discuss in detail in Section~\ref{sec:cp}, maintains a partition of the real line that restrains the possible locations of optimal intervals that are yet to arrive in the stream.  
We target these locations individually in the subsequent run by initiating additional runs of $\mathcal{CP}$ on restricted domains where we expect to find many of these optimal intervals. 

Our approach relies on a property of the $\mathcal{CP}$ algorithm that, at first glance, seems relatively insignificant. As proved by Cabello and P\'{e}rez-Lantero, the $\mathcal{CP}$ algorithm produces a solution of size at least $(|OPT| + 1)/2$, and thus only has an approximation factor of $2$ in an asymptotic sense. Consequently, if $OPT$ is a small constant then the algorithm achieves an approximation factor strictly below $2$. We exploit this property in that we execute the additional runs of $\mathcal{CP}$ on small domains where we expect to find only a small constant number of optimal intervals, see Section~\ref{sec:alg-arbitrary-length} for further details.

Our $(2.5 - \varepsilon)$-approximation lower bound for arbitrary-length intervals is also achieved via a reduction to a hard problem in one-way communication complexity. However, instead of exploiting the hardness of the two-party problem \textsf{Index} as in the unit-lengths case, we use the three-party problem $\textsf{Chain}_3$ introduced by Cormode et al. \cite{Cormode2018IndependentSI} instead. In $\textsf{Chain}_3$, the first two parties and the last two parties hold separate \textsf{Index} instances $(X_1, J_1), (X_2, J_2) \in \{0, 1\}^n \times [n]$ that are correlated in that they have the same answer bit, i.e., $X_1[J_1] = X_2[J_2] =: x$, and the objective for the third party is to determine the bit $x$. $\textsf{Chain}_3$ also requires a message of size $\Omega(n)$ to be solved. 
Similar to the unit-length case, the first two parties introduce clique gadgets based on the bit-strings $X_1$ and $X_2$, and the third party introduces additional crucial intervals. The strength of using $\textsf{Chain}_3$ is that, if the answer bit is zero, then the crucial intervals corresponding to $X_1[J_1]$ and $X_2[J_2]$ of all clique gadgets are missing, while if the answer bit is one then all of these intervals are present.
The method thus allows us to work with multiple clique gadgets instead of only a single one, which we exploit to obtain a stronger lower bound. See Section~\ref{sec:lb-arb-length} for details.

\subparagraph{Further Related Work} \label{app:further-related-work}
Crouch et al. \cite{cms13} initiated the study of graph problems in the sliding window model (recall that \textsf{Interval Selection} is an independent set problem on interval graphs). They showed that, similar to the streaming model, there exist sliding window algorithms that use space $\OrderT(n)$ for deciding \textsf{Connectivity} and \textsf{Bipartiteness}, where $n$ is the number of vertices in the input graph. They also gave positive results for the computation of cut-sparsifiers, spanners and minimum spanning trees, and they initiated the study of the \textsf{Maximum Matching} problem in the sliding window model (see below).

The smooth histogram technique has been successfully applied for designing sliding window algorithms for graph problems, and the state-of-the-art sliding window algorithms for \textsf{Maximum Matching} and \textsf{Minimum Vertex Cover} rely on the smooth histogram technique.  

For \textsf{Maximum Matching}, a $2$-approximation with space $\OrderT(n)$ can easily be achieved in the streaming model by running the \textsc{Greedy} matching algorithm, and the smooth histogram method immediately yields a $(4+\varepsilon)$-approximation sliding window algorithm when built on \textsc{Greedy}. Crouch et al. \cite{cms13} observed that the resulting algorithm can be analyzed more precisely and showed that it actually yields a $(3+\varepsilon)$-approximation sliding window algorithm. Regarding the weighted version of the \textsf{Maximum Matching} problem, the smooth histogram technique immediately yields a $(4+\varepsilon)$-approximation using the $(2+\varepsilon)$-approximation streaming algorithm by \cite{ps19}, and, again, as proved by Biabani et al. \cite{Biabani2021MaximumWeightMI}, the analysis can be tailored to the \textsf{Maximum Matching} problem to establish an approximation factor of $3.5+\varepsilon$ without changing the algorithm. Alexandru et al. \cite{adkk23} then improved the approximation factor to $3+\varepsilon$ by running the smooth histogram algorithm with a slightly different objective function.

Regarding the \textsf{Minimum Vertex Cover} problem, a smooth histogram-based algorithm is known to yield an approximation factor of $(3+\varepsilon)$ \cite{s21}, improving over previous work \cite{kr22}. 

\subparagraph{Outline}
In Section~\ref{sec:prelim}, we give notation, provide some clarification on the sliding window model, and introduce hard communication problems that we rely on for proving our lower bound results. Then, in Section~\ref{sec:unit-length}, we give our algorithm and lower bound for the case of unit-length intervals, and in Section~\ref{sec:arbitrary-length}, we give our algorithm and lower bound for arbitrary-length intervals. Finally, we conclude in Section~\ref{sec:conclusion} with open problems.

\section{Preliminaries} \label{sec:prelim}
For a set of intervals $\mathcal{I}$, we denote by $OPT(\mathcal{I})$ an independent subset of $\mathcal{I}$ of maximum size. We also apply $OPT(.)$ to substreams of intervals and to data structures that store intervals.

\subparagraph*{Sliding Window Algorithms} Throughout the document, we denote by $L$ the size of the sliding window, and we assume that $L$ is large enough, i.e., larger than a suitably large constant. For two streams of intervals $A,B$ we denote the stream that is obtained by concatenating $A$ and $B$ simply by $AB$, i.e., we omit a concatenation symbol. Furthermore,
for simplicity, we assume that the space required to store an interval is $O(1)$. However, if instead $k$ bits are accounted for storing an interval then the space complexities of our algorithms need to be multiplied by $k$.

\subparagraph*{Communication Complexity} As it is standard in the data streaming literature, our space lower bounds are proved via reductions to problems in the {\em one-way communication setting}. In this setting, multiple parties $P_1, P_2, \dots, P_k$ each hold a portion of the input data and communicate in order to solve a problem. Communication is {\em one-way}, i.e., party $P_1$ sends a message to $P_2$, who in turn sends a message to $P_3$. This continues until party $P_k$ has received a message from party $P_{k-1}$ and then outputs the result of the computation. The parties can make use of public and private randomness and need to report a correct solution with probability $2/3$. We refer the reader to  \cite{kn96} for an introduction to communication complexity.

We will exploit the hardness of the two-party communication problem $\textsf{Index}_n$, where we denote the first party by Alice and the second party by Bob, and the $k$-party communication problem $\textsf{Chain}_k$, which was recently introduced by Cormode et al. \cite{Cormode2018IndependentSI}.

\noindent \begin{center}\fbox{ 
\begin{minipage}{0.97\textwidth} 
\textbf{$\textsf{Index}_n$:}
\begin{itemize}
 \item Input: Alice holds a bit-string $X \in \{0, 1\}^n$, and Bob holds an index $J \in [n]$.
 \item Output: Bob outputs $X[J]$.
\end{itemize}
\end{minipage}} \end{center}
It is well-known that solving $\textsf{Index}_n$ requires Alice to send a message of size $\Omega(n)$.
\begin{theorem}[e.g. \cite{Jayram2008TheOC}]\label{thm:index}
Every randomized constant-error one-way communication protocol for $\textsf{Index}_n$ requires a message of size $\Omega(n)$.
\end{theorem}
The problem $\textsf{Chain}_k(n)$ can be regarded as chaining together $k-1$ instances of $\textsf{Index}_n$, where the instances are correlated in that they are guaranteed to have the same output.

\begin{center}
\fbox{\begin{minipage}{0.97\textwidth}
\textbf{\textsf{Chain$_k(n)$}:}
\begin{itemize}
 \item Input: For $1 \le i \le k-1$,  player $P_i$ receives a bitvector $X_i \in \{0,1\}^n$.
 Additionally, for any $2 \le i \le k$ player $P_i$ receives an index $J_{i-1} \in [n]$. The inputs are correlated such that
 $$X_1[J_1] = X_2[J_2] = \dots = X_{k-1}[J_{k-1}] = x \in \{0, 1\}  \ . $$ 
 \item  Output: Player $P_k$ outputs $x$.
\end{itemize}
\end{minipage}} \end{center}
Sundaresan \cite{sundaresan2024optimal} recently settled the communication complexity of $\textsf{Chain}_k(n)$, improving over the previous lower bounds by Cormode et al. \cite{Cormode2018IndependentSI} and Feldman et al. \cite{fnsz20}:

\begin{theorem}[\cite{sundaresan2024optimal}] \label{thm:chain-k}
    Every constant-error one-way communication protocol that solves $\textsf{Chain}_k(n)$ requires at least one message of size $\Omega(n/k)$.
\end{theorem}

\section{Unit-length Intervals} \label{sec:unit-length}
In this section, we give our sliding window algorithm (Subsection~\ref{sec:unit-length-ub}) 
and our lower bound (Subsection~\ref{sec:unit-length-lb}) 
for unit-length intervals.

\subsection{Sliding Window Algorithm for Unit-length Intervals}\label{sec:unit-length-ub}
We now describe our algorithm for unit-length intervals.

\begin{algorithm}[ht]
 \caption{Sliding window algorithm for \textsf{Interval Selection} on unit-length intervals}
 \label{alg:UnitSlidingAlg}
 \vspace{2mm}
 \textbf{Input:}  Stream $S$ of unit-length intervals, 
  window length $L$
 \vspace{2mm}
  \hrule
\vspace{2mm}
\textbf{Initialization:}
\begin{algorithmic}[1]
\State   $\texttt{latest} \leftarrow \emptyset$ the indexed set of stored intervals 
\end{algorithmic}
\vspace{2mm} 
\hrule
\vspace{2mm}
 \textbf{Streaming:}
\begin{algorithmic}[1]
\While{an interval $I = [r,r+1]$ is revealed, for some real number $r$} 
 \State{$\texttt{latest}(\lfloor r \rfloor) \leftarrow I$ }

  \If{$\exists \ J' = [r', r'+1] \in \texttt{latest}$ that has expired} \label{algline:latestdeletebegin}
   \State $\texttt{latest}( \lfloor r' \rfloor)\leftarrow \emptyset$ \label{algline:latestdeleteend}
  \EndIf

\EndWhile
\end{algorithmic}
\vspace{2mm}
\hrule
\vspace{2mm} 
\textbf{Post-processing:}
\begin{algorithmic}[1]
   \State{Return $OPT(\texttt{latest})$} 
\end{algorithmic}
\end{algorithm}

Our algorithm is simple: For each integer $r$, the algorithm maintains in $\texttt{latest}(r)$ the latest interval of the current sliding window with its left boundary in $[r, r+1)$. The key observation, which was also used in \cite{Bakshi2019WeightedMI}, is that the intervals $\{ \texttt{latest}(r) \ : \ r \text{ even} \}$ and $\{ \texttt{latest}(r) \ : \ r \text{ odd} \}$ form independent sets, and one of these sets constitutes a $2$-approximation.

\begin{theorem}
\label{thm:ub-unit-length}
Algorithm \ref{alg:UnitSlidingAlg}  is a deterministic $2$-approximation sliding window algorithm for \textsf{Interval Selection} on unit-length intervals that, at any moment, uses $\Order(|OPT|)$ space, where $OPT$ is a maximum independent set of intervals in the current sliding window.
\end{theorem}

\begin{proof}
We will first prove that  Algorithm \ref{alg:UnitSlidingAlg} indeed computes a $2$-approximation, and then argue that the algorithm satisfies the memory requirements.

We call a unit-length interval $I$ active if it is included in the current sliding window (one of the $L$ most recent intervals of the stream).
Otherwise, we say that $I$ is expired.

Let $OPT$ be a maximum independent set in the current sliding window and let $ALG$ be the independent set reported by Algorithm \ref{alg:UnitSlidingAlg}.
Define the indexed set $\texttt{latest}$ as in the algorithm.

\subparagraph{Approximation} We will show that
\begin{align} \label{eqn:320}
   |OPT| \le |\texttt{latest}| \le 2 \cdot |ALG| 
\end{align} 
holds, which then establishes the approximation factor of the sliding window algorithm.

First, we will prove $|OPT| \le |\texttt{latest}|$ holds. To this end, we will show that the function $f : OPT \to \texttt{latest}$ defined as $f([x,x+1]) = latest(\lfloor x \rfloor)$ is injective.

We will first argue that $f$ is well-defined in that  $\texttt{latest}(\lfloor x \rfloor)$ exists, for every $[x, x+1] \in OPT$.
Indeed, by inspecting the algorithm, when $I :=[x,x+1] \in OPT$ arrives in the stream, $\texttt{latest}(\lfloor x \rfloor)$ is set to $I$, and, in particular,
while $I$ is active, $\texttt{latest}(\lfloor x \rfloor)$ is never set to $\emptyset$. 
It may, however, happen that it is replaced with an interval which appeared after $I$.
In both cases, $f$ is well-defined.

To see that $f$ is injective, observe that for any two intervals in $OPT$, since these intervals are independent and of unit-length, the integer parts of their left endpoints are distinct. Hence, $f(I_1) \neq f(I_2)$, for any two distinct intervals $I_1, I_2 \in OPT$.

Since $f$ is well-defined and injective, we obtain that $|OPT| \le |\texttt{latest}|$, which thus proves the first inequality of Inequality~\ref{eqn:320}.
It remains to prove the second, i.e.,  that $|\texttt{latest}| \le 2 \cdot |ALG|$ also holds.

To see this, observe that, for two integers $x \neq y$ of the same parity, $\texttt{latest}(x)$ and $\texttt{latest}(y)$ (if they exist) are independent. This is because $|y - x| \ge 2$ and the intervals have unit-length.
By the pigeonhole principle, there are at least $\frac{|\texttt{latest}|}{2}$ intervals where their indices inside $\texttt{latest}$ have the same parity, which implies that 
$|ALG| \ge \frac{|\texttt{latest}|}{2}$. 

\subparagraph{Space} The algorithm stores $|\texttt{latest}|$ intervals in the current sliding window. 
Then, as proved above, we have  $|\texttt{latest}| \le 2|ALG| \le 2 |OPT|$, which implies that the space used by the algorithm is $\Order(|OPT|)$.

\end{proof}

\subsection{Space Lower Bound} \label{sec:unit-length-lb}

We now show that sliding window algorithms that use space $o(L)$ cannot compute a $(2-\varepsilon)$-approximation to \textsf{Interval Selection} on unit-length intervals, for any $\varepsilon > 0$. Recall that, in the streaming model, a $\frac{3}{2}$-approximation can be computed with space $O(|OPT|)$.

\begin{theorem}
\label{thm:lb-unit-length}
    Let $\varepsilon > 0$ be any small constant. Then, any algorithm in the sliding window model that computes a $(2 - \varepsilon)$-approximate solution to \textsf{Interval Selection} on unit-length intervals with probability at least $2/3$ requires a memory of size $\Omega(L)$.
\end{theorem}

\begin{proof}
Let $\mathcal{A}$ be a  sliding window algorithm for \textsf{Interval Selection} on unit-length intervals with approximation factor  $2 - \varepsilon$, for some $\varepsilon > 0$.

We will show how $\mathcal{A}$ can be used in order to obtain a communication protocol for $\textsf{Index}_{L-2}$. 

To this end, let $(X, J) \in \{0,1\}^{L-2} \times [L-2]$ be Alice and Bob's input to $\textsf{Index}_{L-2}$. The two players proceed as follows:

\begin{itemize} 
 \item \textbf{Alice:} Alice feeds the intervals $I_1, I_2, \dots, I_{L-2}$ into $\mathcal{A}$ (in the given order), where
\begin{align*}
    I_i = \begin{cases}
        [\frac{i}{2L-1}, 1 + \frac{i}{2L-1}], & \text{if } X[i] = 1 \ , \\
        [1 - \frac{i}{L^2}, 2 - \frac{i}{L^2}], & \text{if } X[i] = 0 \ . 
    \end{cases}
\end{align*}
 Alice then sends the memory state of $\mathcal{A}$ to Bob.
 \vspace{0.3cm}
 
 \item \textbf{Bob:} Using Alice's message, Bob continues the execution of $\mathcal{A}$ and feeds the interval $$I_{L-1} = \left[1 + \frac{J}{2L-1} + \frac{1}{(2L-1)^2}, 2 + \frac{J}{2L-1} + \frac{1}{(2L-1)^2}\right]$$ into $\mathcal{A}$. 
 Bob also adds the intervals $I_i = [\frac{i}{2L-1}, 1 + \frac{i}{2L-1}]$ to $\mathcal{A}$, for $L \le i \le L + J - 1$ in order to make the sliding window of the algorithm $\mathcal{A}$ advance.
 Bob computes $\mathcal{A}$'s output in the latest sliding window consisting of the intervals defined as $\mathcal{S} = \{I_i | J \le i \le L+J-1\}$.

\end{itemize}

This construction is illustrated in Figure~\ref{fig:unitlowerbound}. 

\begin{figure}[ht]
\centering
\begin{tikzpicture}[scale=1.5]

\draw[red] (0,0) node[anchor=east] {$I_1$}  -- (2,0);
\draw[thick,dotted] (1,0) --  (1.3,0.4);
\draw[red,dashed] (1.8,2.2) node[anchor=east] {$I_2$}  -- (3.8,2.2);
\draw[thick, dotted] (2.79,2.2) -- (2.81,2.8);
\draw (0.3,0.4)  node[anchor=east] {$I_J$} -- (2.3,0.4);
\draw (2.4,0.4)  -- (4.4,0.4) node[anchor=west] {$I_{L-1}$};
\draw (0.6,0.7) node[anchor=east] {$I_{J+1}$} -- (2.6,0.7);

\draw[thick, dotted] (1.6,0.7) -- (2,1.3);
\draw (1.85,2.8) [dashed] node[anchor=east] {$I_{L-2}$}  -- (3.85,2.8);
\draw (1.9,3) [dashed] node[anchor=east] {$I_L$}  -- (3.9,3);
\draw (1,1.3) node[anchor=east] {$I_{L+J-2}$} -- (3,1.3);
\draw (1.2,1.6) node[anchor=east] {$I_{L+J-1}$} -- (3.2,1.6);

\end{tikzpicture}
\caption{This figure illustrates the instances created by Alice and Bob in the proof of Theorem~\ref{thm:lb-unit-length} for an instance of $\textsf{Index}_{L-2}$ with $X[J] = 1$. The dashed intervals on the upper part correspond to the zero elements of the bitvector $X$. The red intervals $I_1$, $I_2$ correspond to expired intervals. $I_J$ is the only non-expired interval disjoint with the special interval $I_{L-1}$. Since $X[J] = 1$, the optimal solution is of size $2$. If $X[J]$ was equal to $0$ then the interval $I_J$ would not be disjoint with $I_{L-1}$, and, thus, an optimal solution would be of size $1$. }
\label{fig:unitlowerbound}
\end{figure}

We observe that if $X[J] = 1$ then $|OPT(\mathcal{S})| = |\{I_J, I_{L-1}\}| = 2$, while if $X[J] = 0$ then $|OPT(\mathcal{S})| = 1$. Since $\mathcal{A}$ has an approximation factor of $2-\varepsilon$, $\mathcal{A}$ needs to report the unique solution of size $2$ if $X[J] = 1$, and a solution of size $1$ when $X[J] = 0$. Bob can thus distinguish between the two cases and solve $\textsf{Index}_{L-2}$.

Since the protocol solves $\textsf{Index}_{L-2}$, 
by Theorem~\ref{thm:index}, the protocol must use a message of size $\Omega(L)$. The protocol's message is $\mathcal{A}$'s memory state, and, hence, $\mathcal{A}$ must use space $\Omega(L)$.

\end{proof}

\section{Arbitrary-length Intervals}\label{sec:arbitrary-length}
In this section, we give our $(\frac{11}{3}+\varepsilon)$-approximation sliding window algorithm and our $(\frac{5}{2}-\varepsilon)$-approximation lower bound for \textsf{Interval Selection} on arbitrary-length intervals. 

\subsection{\texorpdfstring{$(\frac{11}{3}+\varepsilon)$}{}-approximation Sliding Window Algorithm}
Our algorithm is obtained by running multiple instances of the Cabello and P\'{e}rez-Lantero streaming algorithm for \textsf{Interval Selection} on arbitrary-length intervals \cite{cp17}. In the following, we abbreviate the algorithm by $\mathcal{CP}$. Since we employ various properties of the $\mathcal{CP}$ algorithm, we discuss the $\mathcal{CP}$ algorithm in Subsection~\ref{sec:cp}. We use the $\mathcal{CP}$ algorithm in the context of the smooth histogram technique, which we discuss in Subsection~\ref{sec:smooth-histogram}. Finally, we give our sliding window algorithm and its analysis in Subsection~\ref{sec:alg-arbitrary-length}.

\subsubsection{Cabello and P\'{e}rez-Lantero Algorithm \label{sec:cp}}

For an interval $I = [a,b]$, we define $\text{left}(I) = a$ and $\text{right}(I) = b$.

The Cabello and P\'{e}rez-Lantero algorithm is depicted in Algorithm~\ref{alg:cabello}.

The listing of the algorithm uses the auxiliary functions $\text{left}(.)$ and $\text{right}(.)$, which return the left and right delimiters of an interval, respectively. 

\begin{algorithm}[ht]
 \caption{Cabello and P\'{e}rez-Lantero Algorithm ($\mathcal{CP}$)}
 \label{alg:cabello}
 \vspace{2mm}
 \textbf{Input:} A stream $S$ of  intervals 
 \vspace{2mm}
  \hrule
\vspace{2mm}
\textbf{Initialization:}
\begin{algorithmic}[1]
\State $\mathcal{R} \leftarrow \{\mathbb{R}\}$ a region partition
\State $\text{leftmost}(\mathbb{R}) = \text{rightmost}(\mathbb{R})  = \mathbb{R}$
\end{algorithmic}
\vspace{2mm} 
\hrule
\vspace{2mm}
 \textbf{Streaming:}
\begin{algorithmic}[1]
\While{an interval $I$ is revealed}
 \If{there exists $R = [a,b) \in \mathcal{R}$ such that $I \subseteq R$}
  \If{$I \cap \text{leftmost}(R) \cap \text{rightmost}(R) \neq \emptyset$}
    \State{\textbf{if} $\text{right}(I) \le \text{right}(\text{leftmost}(R))$ \textbf{then} leftmost$(R) \leftarrow I$}
    \State{\textbf{if} $\text{left}(I) \ge \text{left}(\text{rightmost}(R))$ \textbf{then} rightmost$(R) \leftarrow I$}     
   \Else
     \If{$\text{right}(I) \le \text{left}(\text{rightmost}(R)) $}
      \State{Let $R_1 \leftarrow (a,\text{right}(I)]$}
      \State{Let $R_2 \leftarrow (\text{right}(I)),b]$}
      \State{leftmost$(R_1) \leftarrow I$, rightmost$(R_1) \leftarrow I$}      
      \State{leftmost$(R_2) \leftarrow \text{rightmost}(R)$, rightmost$(R_2) \leftarrow \text{rightmost}(R)$}
     \Else
      \State{Let $R_1 \leftarrow [a,\text{left}(I))$}
      \State{Let $R_2 \leftarrow [\text{left}(I)),b)$}
      \State{leftmost$(R_2) \leftarrow I$, rightmost$(R_2) \leftarrow I$}
      \State{leftmost$(R_1) \leftarrow \text{leftmost}(R)$,rightmost$(R_1) \leftarrow \text{leftmost}(R)$}
      \EndIf
    \State Insert $R_1,R_2$ into $\mathcal{R}$
    \State Remove $R$ from $\mathcal{R}$
    \EndIf
    \EndIf
 \label{streaming:initend}
\EndWhile
\end{algorithmic}
\vspace{2mm}
\hrule
\vspace{2mm} 
\textbf{Post-processing:}
\begin{algorithmic}[1]
   \State Return $\{ \text{leftmost}(R) | R \in \mathcal{R}\}$ 
\end{algorithmic}
\end{algorithm}

The key idea behind the algorithm is to maintain a partition $\mathcal{R}$ of the real line $\mathbb{R}$ that we refer to as a region partition. Initially, the algorithm starts with the single region $\mathcal{R} = \{ \mathbb{R} \}$, and as the algorithm proceeds, the real line is partitioned into half-open intervals. This is achieved as follows. Arriving intervals that cross a region boundary are ignored. Consider thus an arriving interval $I$ that lies entirely within a region. In each region, the algorithm stores the left-most (the interval with the left-most right delimiter) and right-most (the interval with the right-most left delimiter) intervals within the region that it has observed thus far. If the interval $I$ together with either the left-most or the right-most interval of the region forms an independent set of size two then the region is split into two regions and the left-most and right-most intervals are updated accordingly. Otherwise, if $I$ intersects with both the left-most and right-most intervals of the region then $I$ is only used to potentially replace the left-most and/or right-most intervals of the region.

Some key properties of the algorithm that we will reuse in this work are summarized in Figure~\ref{fig:key-cp} (see \cite{cp17} for proofs). 

\begin{figure}[ht]
\noindent 
 \fbox{
\begin{minipage}{0.95 \textwidth}

\begin{enumerate}
 \item[\textbf{C1}] For each reigon $R \in \mathcal{R}$, the $\mathcal{CP}$ algorithm stores at least one (and at most two) intervals and the input instance is such that there are no two disjoint intervals that lie within region $R$. \label{prop:cabellowindowinvariant}
 \item[\textbf{C2}] The $\mathcal{CP}$ algorithm outputs a solution of size $|\mathcal{R}|$, i.e., one interval per region. Furthermore, we have that $|\mathcal{R}| \ge \frac{|OPT|+1}{2}$, i.e., the algorithm has an approximation factor slightly better than $2$. \label{prop:cabelloapproximation}
 \item[\textbf{C3}] The algorithm uses space $O(|OPT|)$.
\end{enumerate}
\end{minipage}}
\caption{Key Properties of the $\mathcal{CP}$ Algorithm. \label{fig:key-cp}}
\end{figure}

Besides these properties, we require another property that allows us to employ the algorithm in the context of the smooth histogram technique:

\begin{lemma} \label{lem:monotonic}
 The $\mathcal{CP}$ algorithm is monotonic, i.e., for any two streams of intervals $A,B$ we have that
 $$|\mathcal{CP}(A)| \le |\mathcal{CP}(A B)| \ . $$
\end{lemma}
\begin{proof}
The output produced by the $\mathcal{CP}$ algorithm consists of one interval per region. The lemma then follows since, by construction, the number of regions cannot decrease. 
\end{proof}

\subsubsection{The Smooth Histogram Technique \label{sec:smooth-histogram}}
The $\mathcal{CP}$ algorithm can be employed in the context of the smooth histogram method to yield a $(4+\varepsilon)$-approximation sliding window algorithm for \textsf{Interval Selection} for arbitrary-length intervals that uses space $\OrderT(|OPT|)$. This is achieved as follows (see Algorithm~\ref{alg:smooth-histogram}):

\begin{algorithm}
 \begin{algorithmic}[1]
  \While{an Interval $I$ is revealed}
   \State{Create new instance of the $\mathcal{CP}$} algorithm
    \State{Feed $I$ into all $\mathcal{CP}$ instances that are currently running}    
    \State{\textbf{Clean-up:} \\  
    $\quad \quad$ Remove oldest run of $\mathcal{CP}$ if it has expired}  \\
    $\quad \quad$ Let $\mathcal{CP}_1, \mathcal{CP}_2, \dots$ be the runs sorted in increasing order with respect to their starting positions. \\
    $\quad \quad$ Fix $i = 1$. \\
    $\quad \quad$ While $\mathcal{CP}_i$ exists, do  the following: \\
    
    $\quad \quad \quad$ Find maximal $j > i$ such that $|\mathcal{CP}_i| \le (1+\varepsilon)|\mathcal{CP}_{j}|$ \\
    $\quad \quad \quad$ If such $j$ exists, delete all the instances of $\mathcal{CP}$ strictly between $i$ and $j$. \\    
    $\quad \quad \quad$ Renumber the runs sorted in increasing order with respect to their starting positions.\\
    $\quad \quad \quad$ Increase $i$ by 1.   
    \State{\textbf{output} solution of oldest run}
  \EndWhile
 \end{algorithmic}
 \caption{Smooth Histogram Technique applied to the $\mathcal{CP}$ algorithm \label{alg:smooth-histogram}} 
\end{algorithm}

Upon the arrival of a new interval, Algorithm~\ref{alg:smooth-histogram} first creates a new run of the $\mathcal{CP}$ algorithm and feeds the new interval into all currently running copies of 
$\mathcal{CP}$. The method relies on a clever way of deleting unnecessary runs.
First, we delete a potentially expired run, i.e., a run which contains an interval that appeared before the start of the current sliding window. Then, for each run $\mathcal{CP}_i$, we take the latest run $\mathcal{CP}_j$ such that $|\mathcal{CP}_i| \le (1+\varepsilon) |\mathcal{CP}_j|$, provided such $j > i$ exists, and we delete every run appearing later than $\mathcal{CP}_i$ but earlier than $\mathcal{CP}_j$. 

Consider the moment after a clean-up took place, and let us denote the stored runs by $\mathcal{CP}_1, \dots, \mathcal{CP}_{\ell}$. The clean-up rule implies that the stored runs have the properties depicted in Figure~\ref{fig:smooth-histogram}.

\begin{figure}[H]
\noindent 
 \fbox{
\begin{minipage}{0.95\textwidth}

\begin{enumerate}
    \item[\textbf{S1}] For every $i \le \ell - 2$: $|\mathcal{CP}_i| \ge  (1+\varepsilon) |\mathcal{CP}_{i+2}|$, and \label{prop:sh-1}
    \item[\textbf{S2}] Either $|\mathcal{CP}_i| \le  (1+\varepsilon) |\mathcal{CP}_{i+1}|$ holds, or, if $|\mathcal{CP}_i| \ge  (1+\varepsilon) |\mathcal{CP}_{i+1}|$ then the starting positions of run $i$ and $i+1$ differ by only a single interval. \label{prop:sh-2}
\end{enumerate}
\end{minipage} }
 \caption{Key Properties of the Smooth Histogram Technique. \label{fig:smooth-histogram}}
\end{figure}

Property~\textbf{S1} implies that there are at most $O(\log_{1+\varepsilon}(L))$ active runs of $\mathcal{CP}$ and thus the space of Algorithm~\ref{alg:smooth-histogram} is at most a factor $O(\log_{1+\varepsilon}(L))$ larger than the space used by $\mathcal{CP}$. 

Property~\textbf{S2} implies that either consecutive runs differ by at most a $1+\varepsilon$ factor in solution size or are such that their starting times differ by only a single interval. 

We now provide a proof that allows us to see that Algorithm~\ref{alg:smooth-histogram} is a $(4+2 \cdot\varepsilon)$-approximation algorithm for \textsf{Interval Selection} for arbitrary-length intervals. This proof will establish insight into how the analysis of our more involved $(\frac{11}{3}+\varepsilon)$-approximation algorithm is conducted and thus serves as a warm-up.

\begin{theorem}
\label{thm:smooth-histogram}
Algorithm~\ref{alg:smooth-histogram} is a $(4+ 2 \cdot \varepsilon)$-approximation sliding window algorithm for \textsf{Interval Selection} on arbitrary-length intervals that uses space $\OrderT(|OPT|)$, where $OPT$ denotes an optimal solution in the current sliding window.
\end{theorem}
\begin{proof}
Recall that the output of Algorithm~\ref{alg:smooth-histogram} is the output of the oldest run of an instance of $\mathcal{CP}$, and let us denote this run by $\mathcal{CP}_1$. First, if the start position of  $\mathcal{CP}_1$ coincides with the oldest interval of the current sliding window then $\mathcal{CP}_1$ was run on the entire region and we immediately obtain an approximation factor of $2$ (Property \textbf{C2} of the $\mathcal{CP}$ algorithm). Hence, suppose that this is not the case, and denote the run that has expired most recently by $\mathcal{CP}_0$. Observe that, prior to $\mathcal{CP}_0$ expiring, the runs $\mathcal{CP}_0$ and $\mathcal{CP}_1$ were adjacent. Furthermore, the two runs differ by more than one interval since otherwise the starting position of $\mathcal{CP}_1$ would coincide with the oldest interval of the current sliding window. We now consider the suffix $S$ of intervals in the stream starting at the start position of run $\mathcal{CP}_0$. This suffix $S$ is partitioned into three parts $S = A B C$, where $A$ are the intervals that arrived prior to the starting position of $\mathcal{CP}_1$, $C$ are the intervals that arrived from the moment onward when the runs $\mathcal{CP}_0$ and $\mathcal{CP}_1$ became adjacent (either after a clean-up or they may have been adjacent from the moment onward when $\mathcal{CP}_1$ was created), and $B$ are the intervals between $A$ and $C$ (if there are any). 

We will prove that $|\mathcal{CP}_1| = |\mathcal{CP}(B C)| \ge |OPT(A B  C)| / (4 + 2 \cdot \varepsilon)$, which proves the result since the current sliding window is a suffix of $A B C$. Indeed, we have the following:
\begin{align*}
 |OPT(A  B C)| & = |OPT(A B C) \cap A | \\
 & \quad + |OPT(A B C) \cap (B \cup C)| \\
 & \le |OPT(A)| + |OPT(BC)| \\
 & \le 2 \cdot |\mathcal{CP}(A)| + 2 \cdot |\mathcal{CP}(B C)| & \text{Approx. factor of $\mathcal{CP}$ (Property ~\textbf{C2})}  \\
 & \le 2 \cdot |\mathcal{CP}(AB)|
 + 2 \cdot |\mathcal{CP}(BC)| & \text{Monotonicity, Lemma~\ref{lem:monotonic}}\\
 & \le 2 \cdot (1+\varepsilon) \cdot |\mathcal{CP}(B)|
 + 2 \cdot |\mathcal{CP}(BC)| & \text{Property~\textbf{S2}} \\
 & \le 2 \cdot (1+\varepsilon) \cdot |\mathcal{CP}(BC)|
 + 2 \cdot |\mathcal{CP}(BC)| & \text{Monotonicity, Lemma~\ref{lem:monotonic}} \\
 & \le (4+ 2\varepsilon) \cdot |\mathcal{CP}(BC)| = (4+2 \varepsilon) \cdot |\mathcal{CP}_1| \ . 
\end{align*}
\end{proof}

\subsubsection{Sliding Window Algorithm \label{sec:alg-arbitrary-length}}

We will expand upon the smooth histogram method as described in Algorithm~\ref{alg:arbitrary-length}.
The key idea is to exploit the structure of the regions created by the runs of $\mathcal{CP}$ in the smooth histogram algorithm. 
Based on these regions, we instantiate additional runs that target areas in which we expect to find many optimal intervals. 

\begin{algorithm}

Whenever two runs $\mathcal{CP}_i$ and $\mathcal{CP}_{i+1}$ in Algorithm~\ref{alg:smooth-histogram} become adjacent (either because of the clean-up operation or because a new run was created), proceed as follows:

\vspace{0.15cm}

Denote by $R_1, \dots, R_{\ell}$ the regions created by $\mathcal{CP}_i$ thus far. 

\begin{enumerate}
    \item For each region $R_i$, we initiate a new run of $\mathcal{CP}$, i.e., this run only considers subsequent arriving intervals that lie within $R_i$. \vspace{0.1cm}
    \item For each pair of consecutive regions $R_i, R_{i+1}$, we initiate a new run of $\mathcal{CP}$ that considers all subsequent intervals that lie within the merged region $R_i R_{i+1}$ (the region consisting of the left boundary of $R_i$ and the right boundary of $R_{i+1}$). \vspace{0.1cm}
    \item \textbf{Clean-up:} The additional runs established in Steps~1 and~2 are associated with the run $\mathcal{CP}_{i+1}$, and whenever $\mathcal{CP}_{i+1}$ is deleted then all associated runs are also deleted. \vspace{0.1cm}
    \item \textbf{Output:} The output is generated from the oldest active run of $\mathcal{CP}$ together with its associated runs as in steps $1$ and $2$ (see the proofs of Lemmas~\ref{lem:case-1} and \ref{lem:case-2}).
\end{enumerate}    

 \caption{$(\frac{11}{3}+\varepsilon)$-approximation Algorithm for \textsf{Interval Selection} on arbitrary-length intervals \label{alg:arbitrary-length}}
 \end{algorithm}

We will now proceed and analyse   Algorithm~\ref{alg:arbitrary-length}. 
To this end, we consider any fixed current sliding window.

First, similar to the analysis of Algorithm~\ref{alg:smooth-histogram}, we note that if the starting position of the oldest run of $\mathcal{CP}$, denoted $\mathcal{CP}_1$, coincides with the left delimiter of the sliding window then we immediately obtain a $2$-approximation (by Property \textbf{C2}). Suppose thus that this is not the case. Again, we consider the run $\mathcal{CP}_0$, which is the latest run that has expired and was previously adjacent to $\mathcal{CP}_1$. We also consider the suffix of intervals $S = A B C$, where $A$ are the intervals starting at the starting position of $\mathcal{CP}_0$ and ending before the starting position of $\mathcal{CP}_1$, $C$ are the intervals that occurred after $\mathcal{CP}_0$ and $\mathcal{CP}_1$ became adjacent, and $B$ are the remaining intervals. Let $OPT = OPT(ABC)$,  let $OPT_A = OPT \cap A$, and define $OPT_B$ and $OPT_C$ similarly. Since the current sliding window is a subset of $A  B  C$, we have that an optimal solution in the current sliding window is of size at most $OPT$.

In Algorithm~\ref{alg:arbitrary-length}, we run the smooth histogram algorithm, Algorithm~\ref{alg:smooth-histogram}, with respect to a parameter $\beta > 0$ (i.e., replace parameter $\varepsilon$ in the listing with $\beta$). Then, as proved in Theorem~\ref{thm:smooth-histogram}, we always have at least a $(4+2\beta)$-approximation at our disposal, i.e., 
$$|\mathcal{CP}_1| = |\mathcal{CP}(BC)| \le \frac{|OPT|}{4+2\beta} \ . $$ 
We define $\varepsilon' \ge 0$ such that 
\begin{equation} \label{eqn:123}
    |\mathcal{CP}(BC)| = \frac{|OPT|}{4 + 2\beta-\varepsilon'}  \ .
\end{equation}
In the following, we will argue that if $\varepsilon'$ is close to $0$ then we can find a better solution using the runs associated with $\mathcal{CP}_1$.

Let $\mathcal{R}$  be the regions created by $\mathcal{CP}_0$ at the moment when $\mathcal{CP}_0$ and $\mathcal{CP}_1$ became adjacent, i.e., the regions created by the run $\mathcal{CP}(AB)$. 
By Property \textbf{C2}, we have $\ell := |\mathcal{R}| = |\mathcal{CP}(AB)|$. 
For each region $R_i \in \mathcal{R}$, let $x_i = |OPT_C \cap R_i|$, i.e., the number of optimal intervals in $C$ that lie within the region $R_i$. Furthermore, we define $X := \sum_{i=1}^{\ell} x_i$. 

In the next lemma, we prove that, provided  $\varepsilon'$ is small, the quantity $X$ is necessarily large, i.e., there are many optimal intervals in $C$ that lie within the regions $R_i$. We will later argue that the associated runs with $\mathcal{CP}_1$ can then be used to find many of these.

\begin{lemma}
\label{lemma:sizeofX}
$$X \ge \frac{2 -\varepsilon'}{1+\beta} \cdot  \ell \ . $$
\end{lemma}
\begin{proof}
Observe that $|OPT| \le X + 2 \ell$, since at most $\ell$ intervals of $OPT$ can intersect the region boundaries $\mathcal{R}$, another $\ell$ intervals of $OPT_{A} \cup OPT_B$ can lie within the $\ell$ regions, and the remaining ones are the $X$ intervals of $OPT_C$.

Then, using  Property~\textbf{S2}, Lemma~\ref{lem:monotonic}, and  Inequality~\ref{eqn:123}, we obtain: 

$$\ell = |\mathcal{CP}(AB)| \le (1 + \beta) |\mathcal{CP}(B)| \le (1+\beta) |\mathcal{CP}(BC)| = \frac{(1+\beta)|OPT|}{4 +2\beta -\varepsilon'}  \le   \frac{(1+\beta)(X + 2\ell)}{4 +2\beta -\varepsilon'}  \ ,$$ 

which implies the result.
\end{proof}

Consider now the run $\mathcal{CP}(B)$, which coincides with the run $\mathcal{CP}_1$ until $\mathcal{CP}_0$ and $\mathcal{CP}_1$ became adjacent. Let $B_1$ be the intervals computed by this run that do not intersect the boundaries of $\mathcal{R}$ and let $B_2$ be the intervals computed by this run that intersect the boundaries of $\mathcal{R}$. Then, since $|B_1| + |B_2| =|\mathcal{CP}(B)|$, we have that either $|B_1| \ge \frac{1}{3} |\mathcal{CP}(B)|$ or $|B_2| \ge \frac{2}{3} |\mathcal{CP}(B)|$. We treat both cases separately in Lemmas~\ref{lem:case-1} and \ref{lem:case-2}:

\begin{lemma} \label{lem:case-1}
    Suppose that $|B_1| \ge \frac{1}{3} |\mathcal{CP}(B)|$. Then, using the associated runs of $\mathcal{CP}_1$, we can output a solution of size at least
     $\frac{7  - 3\varepsilon'}{6(1+\beta)} \ell  \ .$
\end{lemma}
\begin{proof}
We call a region $R_i$ {\em good} if it contains an interval from $B_1$. 
We output the solution obtained from the runs of $\mathcal{CP}$ on $R_i$, for all $i$, and if such a run on a good region leads to no intervals (i.e. $x_i = 0$), then we output the interval from $B_1$ instead. Recall that $\mathcal{CP}$ outputs a solution of size $\frac{x_i + 1}{2}$ if $x_i \neq 0$ (Property \textbf{C2}), and we stress here that the additive $+1$ is key for our analysis. 
We thus obtain a solution of size at least:

$$S = \sum_{R_i \mbox{ good}} \max \left\{ \frac{x_i + 1}{2},  1 \right\}  + \sum_{R_i \mbox{ bad}, x_i \neq 0} \frac{x_i+1}{2}  \ .$$

Recall that $\sum^\ell_{i=1} x_i = X \ge \frac{2-\varepsilon'}{1+\beta} \cdot \ell$. 
Hence,
\begin{align*}   
    S  &\ge \sum_{R_i \mbox{ good}} \frac{x_i + 1}{2}  + \sum_{R_i \mbox{ bad}, x_i \neq 0} \frac{x_i+1}{2} \\
    &\ge \frac{X+|B_1|}{2}  & |B_1| \text{ is the number of good regions}\\
    & \ge \frac{2  -\varepsilon'}{2(1+\beta)} \ell + \frac{1}{6}|\mathcal{CP}(B)| & \text{By Lemma \ref{lemma:sizeofX}}\\
    & \ge \frac{2  -\varepsilon'}{2(1+\beta)} \ell+ \frac{\ell}{6(1+\beta)}  & |\mathcal{CP}(B)| \ge  \frac{|\mathcal{CP}(AB)|}{1+\beta} = \frac{\ell}{1+\beta} \text{ by Prop. \textbf{S2}} \\
    & \ge \frac{7  - 3\varepsilon'}{6(1+\beta)} \ell \ .    
\end{align*}
\end{proof}

\begin{lemma}
\label{lem:case-2}
Suppose that $|B_2| \ge \frac{2}{3} |\mathcal{CP}(B)|$. Then, using the associated runs of $\mathcal{CP}_1$, we can output a solution of size at least
    $ \frac{7  - 3\varepsilon'}{6(1+\beta)} \ell \ .$
\end{lemma}

\begin{proof}
 Let $B_2^1 \subseteq$ be the intervals of $B_2$ that lie on the boundary of two regions $R_i R_{i+1}$ where $i$ is even, and let $B_2^2 = B_2 \setminus B_2^1$. Then, either $|B_2^1| \ge \frac{1}{2} |B_2|$  or $|B_2^2| \ge \frac{1}{2}|B_2|$. 

    Suppose that $|B_2^1| \ge \frac{1}{2} |B_2|$. We only analyse this case since the other case is similar.

    We call an even index $i$ {\em good} if there is an interval in $B_2^1$ that lies on $R_i R_{i+1}$. We consider the runs of $\mathcal{A}$ on pairs of regions $R_i R_{i+1}$ where $i$ is even. Then, we find a solution of size:
    
   \begin{align*}
    S &=  \sum_{2k \mbox{ good}} \max \left\{ \frac{x_{2k} + x_{2k+1} + 1}{2}, 1 \right\} + \sum_{\substack{2k \mbox{ bad} \\ x_{2k} + x_{2k+1} \neq 0}} \frac{x_{2k} + x_{2k+1} + 1}{2}\\
     & \ge \sum_{2k \mbox{ good}} \frac{x_{2k} + x_{2k+1} + 1}{2}  + \sum_{\substack{2k \mbox{ bad} \\ x_{2k} + x_{2k+1} \neq 0}} \frac{x_{2k} + x_{2k+1} + 1}{2} \ . 
   \end{align*}   
   Using the identity $X = \sum_{i=1}^\ell x_i$ and that $|B^1_2|$ is the number of good regions and proceeding as in the proof of Lemma~\ref{lem:case-1}, we obtain:
    \begin{align*}      
      S & \ge  \frac{X + |B^1_2|}{2} 
      \ge \frac{2   -\varepsilon'}{2(1+\beta)} \ell + \frac{1}{6}|\mathcal{CP}(B)| \ge \dots \ge 
    \frac{7  - 3\varepsilon'}{6(1+\beta)} \ell \ .
    \end{align*}   
    
\end{proof}

\begin{theorem} \label{thm:ub-arbitrary-length}
  
  For any constant $\delta > 0$, Algorithm~\ref{alg:arbitrary-length} is a $(11/3 + \delta)$-approximation sliding window algorithm for \textsf{Interval Selection} on arbitrary-length intervals that uses space $\OrderT(|OPT|)$.
\end{theorem}
\begin{proof}

The naive smooth histogram method gives us a solution of size  
$$|\mathcal{CP}(BC)| \ge |\mathcal{CP}(B)| \ge \frac{|\mathcal{CP}(AB)|}{1+\beta} \ge \frac{\ell}{1+\beta} \ , $$ 
where we used the monotonicity of $\mathcal{CP}$ (Lemma~\ref{lem:monotonic}) and Property \textbf{S2}.
Using the associated runs, by Lemmas~\ref{lem:case-1} and \ref{lem:case-2}, we get a solution of size at least 
$$ \frac{7 - 3\varepsilon'}{6(1+\beta)}  \ell \ . $$
Since we can output the larger of the two solutions, in the worst case both solutions have the same value, i.e., when:    
$$\frac{\ell}{1+\beta}  =\frac{7  - 3\varepsilon'}{6(1+\beta)}  \ell  \ ,$$  
which implies $\varepsilon' = \frac{1}{3}$. 
The approximation factor thus is for any $\delta > 0$:
$$4+2 \cdot \beta - \varepsilon' + \delta = 11/3 + 2\cdot \beta + \delta \ .$$
Choosing $\beta = \frac{1}{2} \delta$, and rescaling $\delta$ to $\frac{1}{2} \delta$ gives the result.

  As a consequence of Property~\textbf{S1}, as previously established,
  the smooth histogram algorithm uses $\OrderT(|OPT|)$ space. It remains to argue that the runs created in Steps 1 and 2 of Algorithm \ref{alg:arbitrary-length} only increase the space requirements by a constant times $|OPT|$. 
  
  Indeed, for a fixed instance $\mathcal{CP}_i$, all the runs created by Step 1 are pairwise disjoint (they do not store common intervals) so their cumulative space is $O(|OPT|)$ as we assumed the memory required to store an interval is $O(1)$. 
  Similarly, for the runs created by Step 2, an interval appears in at most two such runs.
  So, the cumulative space is  again $O(|OPT|)$.
  Therefore, the total number of intervals stored in the associated runs is at most $O(|OPT|)$, completing the proof.
\end{proof}

The proof of the approximation factor of the algorithm is shown to be tight in Appendix \ref{app:hard-instance}, meaning that the algorithm does not beat the approximation factor of $\frac{11}{3}$.

\subsection{Space Lower Bound}\label{sec:lb-arb-length}
We now give our space lower bound for sliding window algorithms for \textsf{Interval Selection} on arbitrary-length intervals.  Our result is established by a reduction to the three-party communication problem $\textsf{Chain}_3$.

\begin{theorem}\label{thm:lb-arbitrary-length}
Let $\varepsilon > 0$ be any small constant. Then, any algorithm in the sliding window model that computes a $(2.5 - \varepsilon)$-approximate solution to \textsf{Interval Selection} on arbitrary-length intervals requires a memory of size $\Omega(L)$.  
\end{theorem}

\begin{proof}
Let $\mathcal{A}$ be a sliding window algorithm with approximation factor $2.5- \varepsilon$, for some $\varepsilon > 0$, as in the statement of the theorem, and let $n = \frac{L-2}{3}$, where $L$ is the window length. We will argue how $\textsf{Chain}_3(n)$ can be solved with the help of $\mathcal{A}$.

To this end, denote the three parties in the communication problem $\textsf{Chain}_3(n)$ by Alice, Bob, and Charlie. Let $X_1 \in \{0,1\}^n$ be Alice's input, let $X_2 \in \{0,1\}^n$ and $ J_1 \in [n]$ be Bob's input, and let $J_2 \in [n]$ be Charlie's input. The players proceed as follows:

\begin{itemize}
\item \textbf{Alice:} For every $i \in [n]$, Alice feeds the following  intervals into $\mathcal{A}$: 

\noindent \begin{minipage}{0.45\textwidth}
\begin{align*}
    I_1(i) = \begin{cases}
        \left[\frac{i}{3n}, 1 + \frac{i}{3n}\right], & \text{if } X_1[i] = 1 \ , \\
        [-10 - i ,10 + i], & \text{if } X_1[i] = 0 \ . 
    \end{cases}
\end{align*}
\end{minipage} \hfill 
\begin{minipage}{0.45\textwidth}
\begin{align*}
    I_2(i) = \begin{cases}
        \left[2 - \frac{i}{3n}, 3 - \frac{i}{3n}\right], & \text{if } X_2[i] = 1 \ , \\
        [-11 - i ,11 + i], & \text{if } X_2[i] = 0 \ . 
    \end{cases}
\end{align*}
\end{minipage}

\vspace{0.3cm}

The order given is $I_1(1), I_2(1), I_1(2), I_2(2), \dots, I_1(n), I_2(n)$.
We observe that, for every $i \in [n]$,  when $X_1[i] = 1$, the intervals $I_1(i)$ and $I_2(i)$ are disjoint.
Alice sends the memory state of $\mathcal{A}$ to Bob.
 
\vspace{0.3cm}
\item \textbf{Bob:} For every $i \in [n + 2(J_1 - 1)]$, Bob feeds the following interval into $\mathcal{A}$:

\begin{align*}
    I_3(i) = \begin{cases}
        \left[1 + \frac{J_1}{3n} + \frac{1}{6n} + \frac{i-1}{6n^2}, 2 - \frac{J_1}{3n} - \frac{1}{6n} + \frac{i-1}{6n^2}\right], & \text{if } i \le n \text{ and } X[i] = 1 \ , \\
        [-10 - i ,11 + i], & \text{otherwise} \ . 
    \end{cases}
\end{align*}

Let $k \in [n]$. Notice that, for every $j \in [n+2(J_1 - 1)]$, when $X_1[k] = X_2[j] = 1$, we have that $I_3(j)$ is disjoint with both $I_1(k)$ and $I_2(k)$  if and only if $j \le J_1$.
Otherwise, $I_3(j)$ intersects with both $I_1(k)$ and $I_2(k)$.

Bob sends the memory state of $\mathcal{A}$ and $J_2$ to Charlie.

\vspace{0.3cm}
 \item \textbf{Charlie:}
 We denote the interval boundaries of $I_3(i)$ by $a_{I_3(i)}$ and $b_{I_3(i)}$, i.e., $I_3(i) = [a_{I_3(i)}, b_{I_3(i)}]$. Charlie feeds the following two intervals into $\mathcal{A}$: 
 \begin{align*}
 I_{J_2} & = \left[\frac{2a_{I_3(J_2 - 1)} + a_{I_3(J_2)}}{3},\frac{a_{I_3(J_2-1)} + 2a_{I_3(J_2)}}{3}\right] \ , \mbox{ and} \\
 I'_{J_2} & = \left[\frac{2b_{I_3(J_2-1)} + b_{I_3(J_2)}}{3},\frac{b_{I_3(J_2-1)} + 2b_{I_3(J_2)}}{3}\right] \ . 
 \end{align*}
 
 Notice that $I_{J_2}$ intersects all intervals of $I_3(i)$, for all $i < J_2$, while  $I'_{J_2}$ intersects all intervals of $I_3(i)$, for all $i > J_2$.
 
 Using $\mathcal{A}$, Charlie computes the largest independent set $OPT$ of $$\mathcal{I} = \{I_1(k) | J_1 \le k \le n\} \cup \{I_2(k) | J_1 \le k \le n\} \cup \{I_3(k) |  1 \le k \le n + 2 (J_1-1)\} \cup \{I_{J_2}, I'_{J_2}\} \ ,$$
which is possible since $\mathcal{A}$ is a sliding window algorithm and thus able to solve the situation when the intervals $\cup_{1 \le k < J_1} \left(I_1(k) \cup I_2(k)\right)$ have expired.
\end{itemize}

Figure~\ref{fig:variablelowerbound}  provides an illustration of the proof of Theorem~\ref{thm:lb-arbitrary-length}.
\begin{figure}[ht]
\centering
\begin{tikzpicture}[scale=1.3]

\draw (0,0)[red] node[anchor=east] {$I_1(1)$}  -- (2,0);
\draw[thick,dotted,red] (1,0) -- (1.7,0.5);
\draw (0.8,0.5)[red] node[anchor=east] {$I_1(J_1-1)$}  -- (2.8,0.5);
\draw (1.4,0.85)  node[anchor=east] {$I_1(J_1)$} -- (3.4,0.85);
\draw (2.1,1.2)  node[anchor=east] {$I_1(J_1+1)$} -- (4.1,1.2);
\draw[dashed] (3.4,0) -- (3.4,5.3);
\draw[thick, dotted] (3.1,1.2) -- (3.7,1.6);
\draw[dashed] (4.1,0) -- (4.1,5.3);

\draw (2.8,1.6) node[anchor=east] {$I_1(n)$} -- (4.8,1.6);

\draw (5.4,1.6)  -- (7.4,1.6)  node[anchor=west] {$I_2(n)$};
\draw[thick, dotted] (7.1,1.2) -- (6.5,1.6);
\draw (6.1,1.2)   -- (8.1,1.2) node[anchor=west] {$I_1(J_1+1)$};
\draw[dashed] (6.1,0) -- (6.1,5.3);
\draw (6.8,0.85)  -- (8.8,0.85) node[anchor=west] {$I_2(J_1)$};
\draw[dashed] (6.8,0) -- (6.8,5.3);
\draw (7.2,0.5)[red]  -- (9.2,0.5) node[anchor=west] {$I_2(J_1-1)$};
\draw (8,0)[red]  -- (10,0) node[anchor=west] {$I_2(1)$};
\draw[thick,dotted,red] (9,0) -- (8.3,0.5);

\draw  (3.4,2.5)  -- (6.2,2.5);
\draw (4.8,2.5) node[anchor=south] {$I_3(1)$};

\draw[thick,dotted] (4.8,2.5) -- (4.85,3.3);
\draw  (3.6,3.3)  -- (6.3,3.3);
\draw (4.75,3.3) node[anchor=south] {$I_3(J_2-1)$};

\draw  (3.8,3.8)  -- (6.4,3.8);
\draw (5.1,3.8) node[anchor=south] {$I_3(J_2)$};
\draw (3.55,3.8) -- (3.65,3.8);
\draw  (3.6,3.8) node[anchor=south] {$I_{J_2}$};

\draw (6.5,3.8) -- (6.6,3.8);
\draw (6.55,3.8) node[anchor=south] {$I'_{J_2}$};

\draw  (3.9,4.4) -- (6.7,4.4);
\draw (5.25,4.4) node[anchor=south] {$I_3(J_2+1)$};

\draw[thick,dotted] (5.2,4.4) -- (5.35,5.2); 
\draw  (4,5.2) -- (6.8,5.2);
\draw (5.4,5.2) node[anchor=south] {$I_3(n)$};

\end{tikzpicture}
\caption{This figure illustrates the intervals created by Alice, Bob and Charlie in the proof of Theorem \ref{thm:lb-arbitrary-length} for an instance of $\textsf{Chain}_3(n)$ where $n = \frac{L-2}{3}$ with  $X_1[J_1] = X_2[J_2] = 1$.  The red intervals in the figure ($I_1(1),I_1(J_1 - 1)$,$I_2(1)$,$I_2(J_1 - 1)$) correspond to expired intervals. The optimal solution is $\{I_1(J_1), I_2(J_1), I_{J_2}, I'_{J_2}, I_3(J_2)\}$ of size 5 . Otherwise,  if $X_1[J_1] = X_2[J_2] = 0$, then the optimal solution would have been of size 2. All intervals $I_3(i)$ are disjoint from $I_1(J_1)$ and $I_2(J_2)$. However, they intersect with $I_1(J_1 + 1)$ and $I_2(J_2 + 1)$ as emphasized by the vertical dashed lines. Intervals $I_3(i)$ for $n+2(J_1 - 1) \ge i > n$ have been omitted as they do not impact the optimal solution and their only role is to advance the sliding window.}.  
\label{fig:variablelowerbound}
\end{figure}

 The total number of intervals added by the three players is  $3n +2 + 2(J_1-1) = L + 2 (J_1 - 1)$. So, after Charlie's execution $\mathcal{A}$, the incumbent region indeed consists of $\mathcal{I}$.
 
We will argue now that if $X_1[J_1] = X_2[J_2] = 1$ then the optimal solution size is $5$, while if $X_1[J_1] = X_2[J_2] = 0$ then the optimal solution size is $2$.

 Suppose thus that $X_1[J_1] = X_2[J_2] = 1$. Then it is not hard to see that the unique optimal solution is $\{I_1(J_1),I_2(J_1),I_3(J_2),I_{J_2},I'_{J_2}\}$ of size $5$.
 
Next, suppose that $X_1[J_1] = X_2[J_2] = 0$. Notice first that, in this case, $I_1(J_1),I_2(J_1),I_3(J_2)$  intersect with every other interval in the input, so they can only belong to independent sets of size at most $1$.

Also, we have that any interval $I_1(i)$ with $i > J_1$ would block all the intervals $I_3(j)$ for $1 \le j \le n + 2(J_1 - 1)$ and $I_{J_2}$.
So, an interval from $I_1(i)$ with $i > J_1$ can be included in an optimal set of size at most $2$ (either $\{I_1(i), I'_{J_2}\}$ or $\{I_1(i),I_2(j)\}$ for some $j > J_1$).
Similarly, $I_2(i)$ with $i > J_2$ can be included in an optimal set of size at most $2$. 
Furthermore, we can construct from Bob and Charlie's input a solution of size at most $2$ (similar to the $\frac{3}{2} - \varepsilon$ lower bound construction of \cite{ehr16}). The size of an optimal solution is thus in this case $2$.

Recall that $\mathcal{A}$ has an approximation factor of $2.5 - \varepsilon$. Hence, if $X_1[J_1] = X_2[J_2] = 1$ then $\mathcal{A}$ reports a  solution of size at least $3$, thereby distinguishing it from the case when $X_1[J_1] = X_2[J_2] = 0$, which yields an optimal size of $2$.

 Since, by Theorem \ref{thm:chain-k}, $\textsf{Chain}_3(n)$ requires a message of size $\Omega(n)$, and since the protocol solely consists of forwarding the memory state of $\mathcal{A}$, we conclude that $\mathcal{A}$ requires a memory of size $\Omega(n) = \Omega(L)$, which completes the proof.
\end{proof}

\section{Conclusion}\label{sec:conclusion}
In this paper, we initiated the study of the \textsf{Interval Selection} problem in the sliding window model of computation. We gave algorithms and lower bounds for both unit-length and arbitrary-length intervals. In the unit-length case, we gave a $2$-approximation algorithm that uses space $O(|OPT|)$, and we showed that this is best possible in that any $(2-\varepsilon)$-approximation algorithm requires space $\Omega(L)$. In the arbitrary-length case, we gave a $(\frac{11}{3}+\varepsilon)$-approximation algorithm that uses space $\OrderT(|OPT|)$, and we showed that any $(\frac{5}{2}-\varepsilon)$-approximation algorithm requires space $\Omega(L)$. Contrasted with results known from the one-pass streaming setting, our result implies that  \textsf{Interval Selection} in both the unit-length and the arbitrary-length cases is harder to solve in the sliding window setting than in the one-pass streaming setting.

We conclude with two open questions. 

First, the approximation guarantees of our algorithm for arbitrary-length intervals and our respective lower bound do not match. Can we close this gap?

Second, the sliding window model has received significantly less attention for the study of graph problems than the traditional one-pass streaming setting. While from a theoretical perspective, the sliding window model is less clean than the one-pass streaming model, as discussed in the introduction, it is, however, the more suitable model for many applications. We are particularly interested in understanding the differences between the two models. For example, which graph problems can be solved equally well in the sliding window model as in the one-pass streaming setting, and which problems are significantly harder to solve?
    
\section*{Acknowledgement}

We are grateful to Sang-Sub Kim for identifying that the clean-up routine in the smooth histogram method used in an earlier version of this paper did not, as claimed, establish properties \textbf{S1} and \textbf{S2}. This issue has been addressed and corrected in the current version.

\bibliography{ak24}

\begin{thebibliography}{10}

\bibitem{adkk23}
Cezar{-}Mihail Alexandru, Pavel Dvor{\'{a}}k, Christian Konrad, and Kheeran~K.
  Naidu.
\newblock Improved weighted matching in the sliding window model.
\newblock In Petra Berenbrink, Patricia Bouyer, Anuj Dawar, and
  Mamadou~Moustapha Kant{\'{e}}, editors, {\em 40th International Symposium on
  Theoretical Aspects of Computer Science, {STACS} 2023, March 7-9, 2023,
  Hamburg, Germany}, volume 254 of {\em LIPIcs}, pages 6:1--6:21. Schloss
  Dagstuhl - Leibniz-Zentrum f{\"{u}}r Informatik, 2023.
\newblock \href {https://doi.org/10.4230/LIPIcs.STACS.2023.6}
  {\path{doi:10.4230/LIPIcs.STACS.2023.6}}.

\bibitem{Bakshi2019WeightedMI}
Ainesh Bakshi, Nadiia Chepurko, and David~P. Woodruff.
\newblock Weighted maximum independent set of geometric objects in turnstile
  streams.
\newblock In {\em International Workshop and International Workshop on
  Approximation, Randomization, and Combinatorial Optimization. Algorithms and
  Techniques}, 2019.
\newblock URL: \url{https://api.semanticscholar.org/CorpusID:67856291}.

\bibitem{Biabani2021MaximumWeightMI}
Leyla Biabani, Mark de~Berg, and Morteza Monemizadeh.
\newblock Maximum-weight matching in sliding windows and beyond.
\newblock 2021.
\newblock URL: \url{https://api.semanticscholar.org/CorpusID:245276580}.

\bibitem{bO07}
Vladimir Braverman and Rafail Ostrovsky.
\newblock Smooth histograms for sliding windows.
\newblock In {\em 48th Annual {IEEE} Symposium on Foundations of Computer
  Science {(FOCS} 2007), October 20-23, 2007, Providence, RI, USA,
  Proceedings}, pages 283--293. {IEEE} Computer Society, 2007.
\newblock \href {https://doi.org/10.1109/FOCS.2007.55}
  {\path{doi:10.1109/FOCS.2007.55}}.

\bibitem{cp17}
Sergio Cabello and Pablo P{\'{e}}rez{-}Lantero.
\newblock Interval selection in the streaming model.
\newblock {\em Theor. Comput. Sci.}, 702:77--96, 2017.
\newblock \href {https://doi.org/10.1016/j.tcs.2017.08.015}
  {\path{doi:10.1016/j.tcs.2017.08.015}}.

\bibitem{Cormode2018IndependentSI}
Graham Cormode, Jacques Dark, and Christian Konrad.
\newblock Independent sets in vertex-arrival streams.
\newblock {\em ArXiv}, abs/1807.08331, 2018.
\newblock URL: \url{https://api.semanticscholar.org/CorpusID:49907556}.

\bibitem{cms13}
Michael~S. Crouch, Andrew McGregor, and Daniel~M. Stubbs.
\newblock Dynamic graphs in the sliding-window model.
\newblock In Hans~L. Bodlaender and Giuseppe~F. Italiano, editors, {\em
  Algorithms - {ESA} 2013 - 21st Annual European Symposium, Sophia Antipolis,
  France, September 2-4, 2013. Proceedings}, volume 8125 of {\em Lecture Notes
  in Computer Science}, pages 337--348. Springer, 2013.
\newblock \href {https://doi.org/10.1007/978-3-642-40450-4\_29}
  {\path{doi:10.1007/978-3-642-40450-4\_29}}.

\bibitem{ddk23}
Jacques Dark, Adithya Diddapur, and Christian Konrad.
\newblock Interval selection in data streams: Weighted intervals and the
  insertion-deletion setting.
\newblock In {\em Foundations of Software Technology and Theoretical Computer
  Science}, 2023.
\newblock URL: \url{https://api.semanticscholar.org/CorpusID:266192962}.

\bibitem{dgim02}
Mayur Datar, Aristides Gionis, Piotr Indyk, and Rajeev Motwani.
\newblock Maintaining stream statistics over sliding windows: (extended
  abstract).
\newblock In {\em Proceedings of the Thirteenth Annual ACM-SIAM Symposium on
  Discrete Algorithms}, SODA '02, page 635–644, USA, 2002. Society for
  Industrial and Applied Mathematics.

\bibitem{ehr16}
Yuval Emek, Magn{\'{u}}s~M. Halld{\'{o}}rsson, and Adi Ros{\'{e}}n.
\newblock Space-constrained interval selection.
\newblock {\em {ACM} Trans. Algorithms}, 12(4):51:1--51:32, 2016.
\newblock \href {https://doi.org/10.1145/2886102} {\path{doi:10.1145/2886102}}.

\bibitem{fnsz20}
Moran Feldman, Ashkan Norouzi-Fard, Ola Svensson, and Rico Zenklusen.
\newblock The one-way communication complexity of submodular maximization with
  applications to streaming and robustness.
\newblock In {\em Proceedings of the 52nd Annual ACM SIGACT Symposium on Theory
  of Computing}, STOC 2020, page 1363–1374, New York, NY, USA, 2020.
  Association for Computing Machinery.
\newblock \href {https://doi.org/10.1145/3357713.3384286}
  {\path{doi:10.1145/3357713.3384286}}.

\bibitem{Jayram2008TheOC}
T.~S. Jayram, Ravi Kumar, and D.~Sivakumar.
\newblock The one-way communication complexity of hamming distance.
\newblock {\em Theory Comput.}, 4:129--135, 2008.
\newblock URL: \url{https://api.semanticscholar.org/CorpusID:15825208}.

\bibitem{kr22}
Robert Krauthgamer and David Reitblat.
\newblock Almost-smooth histograms and sliding-window graph algorithms.
\newblock {\em Algorithmica}, 84(10):2926--2953, 2022.
\newblock \href {https://doi.org/10.1007/s00453-022-00988-y}
  {\path{doi:10.1007/s00453-022-00988-y}}.

\bibitem{kn96}
Eyal Kushilevitz and Noam Nisan.
\newblock {\em Communication complexity}.
\newblock Cambridge University Press, 1997.

\bibitem{ps19}
Ami Paz and Gregory Schwartzman.
\newblock A (2+{\(\epsilon\)})-approximation for maximum weight matching in the
  semi-streaming model.
\newblock {\em {ACM} Trans. Algorithms}, 15(2):18:1--18:15, 2019.
\newblock \href {https://doi.org/10.1145/3274668} {\path{doi:10.1145/3274668}}.

\bibitem{s21}
Sai Krishna Chaitanya Nalam~Venkata Subrahmanya.
\newblock Vertex cover in the sliding window model.
\newblock Master's thesis, Rutgers, The State University of New Jersey, 2021.

\bibitem{sundaresan2024optimal}
Janani Sundaresan.
\newblock Optimal communication complexity of chained index, 2024.
\newblock \href {http://arxiv.org/abs/2404.07026} {\path{arXiv:2404.07026}}.

\end{thebibliography}

\newpage 
\appendix

\section{Hard Instance for the Analysis of Algorithm \ref{alg:arbitrary-length}} \label{app:hard-instance}

\subsection{Description of the Instance}

We will present a hard instance demonstrating that the analysis of Algorithm \ref{alg:arbitrary-length} is tight. 
We assume that the smooth histogram parameter $\beta$ is set to  $\beta = 0$.
This is a reasonable assumption since the approximation factor of the algorithm approaches the optimal value of $\frac{11}{3}$ when $\beta \rightarrow 0$.

As before, let $\mathcal{CP}_1$ be the oldest still active instance created by the smooth histogram algorithm, and let $\mathcal{CP}_0$ be the expired instance which came before $\mathcal{CP}_1$.
Given  $\mathcal{CP}_0$ and $\mathcal{CP}_1$, we can divide the active stream into successive parts $A,B,C$. 
$A$ represents the intervals that arrived  before the starting position of $\mathcal{CP}_1$.
$C$ represents the intervals that arrived right after the runs $\mathcal{CP}_0$ and $\mathcal{CP}_1$ became adjacent (i.e after all the instances between  $\mathcal{CP}_0$ and $\mathcal{CP}_1$ are deleted).
The intervals arriving after $A$ but before $C$ are denoted as $B$.

The smooth histogram condition then translates to $\mathcal{CP}(AB) = \mathcal{CP}(A)$.

Let $\ell$ be a positive integer divisible by 3.
We will first give the full stream in Algorithm \ref{alg:hardstream} and then explain the purpose of each portion of the stream.

\begin{algorithm}[ht]
 \caption{Hard instance stream $S = ABC$ for Algorithm \ref{alg:arbitrary-length}}
 \label{alg:hardstream}
 \vspace{2mm}
 \textbf{Stream A \newline \newline}
 \textbf{Stream $A_1$} 
 \begin{algorithmic}[1]
 \For {$x$ from $1$ to $\ell$}
  \State Insert $[x + 0.1,x+1]$
  \EndFor

 \end{algorithmic}
 \textbf{Stream $A_2$}
 \begin{algorithmic}[1]
  
  \For {$x$ from $1$ to $\ell$}
  \State Insert $[x + 0.5,x+0.54]$ 
 \EndFor    
 \end{algorithmic}

  \textbf{Stream $A_3$}
 \begin{algorithmic}[1]
 \For {$x$ from $1$ to $\ell$}
  \State Insert $[x + 0.95,x+1.05]$ 
 \EndFor 
 \end{algorithmic}
 \vspace{2mm}
  \hrule
\vspace{2mm}

\textbf{Stream B}
\begin{algorithmic}[1]

\For {$x$ from $1$ to $\ell$}
  \If {$x = 3k + 2$ for integer $k$}
   \State Insert $[x-0.1, x+0.26]$
   \State Insert $[x+0.53, x+0.71]$
   \State Insert $[x+0.9, x+1.1]$
  \EndIf
\EndFor

\end{algorithmic}

\vspace{2mm} 
\hrule
\vspace{2mm}
 \textbf{Stream C \newline \newline}
 \textbf{Stream} $C_1$:
\begin{algorithmic}[1]
\For{$x$ from $1$ to $\ell$}
\If{$x = 3k + 2$ for integer $k$}
\State Insert $[x + 0.06, x + 0.3]$
\State Insert $[x + 0.35, x + 0.75]$
\EndIf
\EndFor
\end{algorithmic}
\textbf{Stream} $C_2$:
\begin{algorithmic}[1]
\For{$x$ from $1$ to $\ell$}
\If{$x = 3k + 2$ for integer $k$}
\State Insert $[x + 0.06, x + 0.15]$
\State Insert $[x + 0.25, x + 0.35]$
\State Insert $[x + 0.55, x + 0.6]$
\State Insert $[x + 0.7, x + 0.8]$
\State Insert $[x + 0.9, x + 0.94]$
\EndIf
\EndFor
\end{algorithmic}
\vspace{2mm}
\end{algorithm}

We call the created regions $[x,x+1)$  with $x = 3k + 2$, for some integer $k$, {\em good}.
Notice that the number of regions created by $\mathcal{CP}(A)$ is $\ell$ while the number of good regions is exactly $\ell/3$ because $\ell$ is divisible by 3.

We will now discuss the instance created in Algorithm~\ref{alg:hardstream}, which is also sketched in Figure \ref{fig:hardisntancesketch}.

\begin{figure}
\centering
\begin{tikzpicture}[scale=1.5]

\draw[very thick] (0.5,2) -- (0.5,-2);
\draw[very thick] (5.5,2) -- (5.5,-2);
\draw[very thick,green] (1.8,2) -- (1.8,-2);
\draw[very thick,green] (3.5,2) -- (3.5,-2);

\draw (-0.5,1) node[anchor=east] {$A_2$};
\draw (-0.5,0.5) node[anchor=east] {$A_3$};
\draw (-0.5,0) node[anchor=east] {$B$};
\draw (-0.5,-0.5) node[anchor=east] {$C_1$};
\draw (-0.5,-1) node[anchor=east] {$C_2$};
\draw (0.4,0.5) -- (0.6,0.5);

\draw[blue] (0.2,0) -- (1.75,0);
\draw[green] (0.75,-0.5) -- (1.8,-0.5);
\draw[green] (0.75,-1) -- (1.05,-1);
\draw[green] (1.65,-1) -- (1.95,-1);
\draw[green] (2.05,-0.5) -- (3.5,-0.5);
\draw (2.4,1) -- (2.7,1);
\draw[blue] (2.6,0) -- (3.45,0);
\draw[green] (2.9,-1) -- (3.1,-1);
\draw[green] (3.4,-1) -- (3.6,-1); 
\draw[green] (5,-1) -- (5.2,-1);
\draw (5.4,0.5) -- (5.6,0.5);
\draw[blue] (5,0) -- (5.8,0);
\end{tikzpicture}
\caption{Illustration of a good region, where the vertical black lines depict its boundaries. Black intervals represent streams $A_2,A_3$. Blue intervals represent the stream $B$. The upper green intervals belong to $C_1$. The bottom green intervals belong to $C_2$. Stream $A_1$ is omitted to keep the illustration simple -- these intervals are responsible for creating the regions of $\mathcal{CP}(A)$).}
\label{fig:hardisntancesketch}
\end{figure}

\begin{itemize}

\item \textbf{Stream $A$}

Stream $A$ has three parts in this order: $A_1, A_2, A_3$.
\begin{itemize} 
\item \textbf{Stream $A_1$} 

This stream is responsible for creating the regions $[x,x+1)$ for any integer $x \in [\ell] \setminus \{1,\ell\}$ and regions $(-\infty, 2),[\ell, \infty)$.

\item \textbf{Stream $A_2$}

This stream inserts intervals inside the regions $[x,x+1)$.
Notice that the interval $[x+0.5,x+0.54] \in A_2$ is completely inside the region $[x,x+1)$ and replaces the interval $[x+0.1, x+1] \in A_1$ as both the leftmost and the rightmost intervals of the region.

\item \textbf{Stream $A_3$} 

This stream inserts intervals that intersect the boundaries of regions created by $A_1$.
The stream contributes to the size of the optimal solution of the overall stream $ABC$.
\end{itemize}
\item \textbf{Stream $B$}

If we execute stream $B$ immediately after stream $A$, the regions created in $A$ will remain unchanged.
In $B$, we only add items inside or intersecting the good regions. 
Consider therefore a good region $R = [x,x+1)$.

The intervals $[x-0.1,x+0.26], [x+0.9,x+1.1]$ are intervals crossing the boundary of $R$. They completely include the boundary intervals of $A_3$ (i.e $[x-0.05, x+0.05]$  from the previous region and $[x + 0.95, x + 1.05]$).
The interval $[x+0.53,x+0.71] \in B$ is an interval inside $R$ that intersects the interval $[x + 0.5, x+0.54] \in A_2$.

\item \textbf{Stream $C$}

The stream $C$ is divided into $C_1,C_2$ and only adds intervals completely included within good regions. Consider therefore a good region $R = [x,x+1)$.
\begin{itemize}
  \item \textbf{Stream $C_1$}
  
The purpose of the stream $C_1$ is to create new regions from the regions of the run $\mathcal{CP}(AB)$.
The new regions created inside $R$ are $[x,  x+0.3), [x + 0.3, x + 0.75)$ and $ [x + 0.75, x + 1)$.

  \item \textbf{Stream $C_2$}
  
The stream $C_2$ contributes to $|OPT \cap C|$ of our instance, where $OPT$ is an independent set of optimal size inside the stream $ABC$.
The interval $[x + 0.06, x + 0.15] \in C_2$ intersects the interval $[x-0.1, x+0.26] \in B$, but it does not intersect the interval $[x - 0.05, x + 0.05] \in A_3$.
The interval $[x + 0.9, x + 0.94] \in C_2$ has similar properties.
The intervals $[x + 0.25, x + 0.35] \in C_2$ and $[x + 0.7, x + 0.8] \in C_2$ intersect with the boundaries of the regions of $\mathcal{CP}(C_1)$, so they will not be saved by the algorithm.
Lastly, the interval $[x + 0.55, x+0.6] \in C_2$ does not intersect with $[x + 0.5, x+0.54] \in A_2$, but it intersects with $[x + 0.53, x + 0.71] \in B$.

\end{itemize}
\end{itemize}

\subsection{Analysis of the Instance}
Here we will prove that the output is indeed a $\frac{11}{3}$ approximation of the optimal solution, therefore proving that our analysis of Algorithm \ref{alg:arbitrary-length} is best possible.

\begin{lemma}

The streams $A,B$ yield $\mathcal{CP}(AB) = \mathcal{CP}(B) = \ell$, hence the smooth histogram condition is obeyed.
\end{lemma}

\begin{proof}
Notice that after the run of stream $A$, we have created the regions $[x,x+1)$ for $x \in [\ell]$ and regions $(-\infty, 2), [\ell, \infty)$.

Now, we consider a good region $R = [x, x+1)$.
The saved interval of $A$ inside $R$ is $[x+0.5, x+0.54]$.

When the stream $B$  arrives, the intervals $[x-0.1,x+0.26]$ and $[x+0.9,x+1.1]$ cross the boundaries of $R$.
The interval $[x + 0.53, x+0.71] \in B$ intersects with the interval $[x+0.5,x+0.54] \in A$, so only the rightmost of the region is changed after the stream $B$ is processed.

Since no new regions are created by $B$, we can argue that $\mathcal{CP}(AB) = \ell$ (the number of regions created by $\mathcal{CP}(A)$).
Furthermore, all the intervals of $B$ are pairwise independent so that $\mathcal{CP}(B) = |B| = \ell$, hence proving the required lemma.

\end{proof}

\begin{lemma}\label{lemma:optsize}
Let $OPT$ be an optimal independent set of stream $ABC$. Then, $|OPT| \ge \frac{11\ell}{3}$.
\end{lemma}

\begin{proof}
Inspecting the intervals given by the streams $A_2,A_3,C_2$, we see that they form an independent set. We have
\begin{itemize}
\item $|A_2| = \ell$,
\item $|A_3| = \ell$, and
\item $ |C_2| = 5 \cdot \frac{\ell}{3}$.
\end{itemize}
Hence, we obtain that $|A_2 \cup A_3 \cup C_2| = \frac{11\ell}{3}$, which implies $|OPT| \ge |A_2 \cup A_3 \cup C_2| = \frac{11\ell}{3}$ as required.
\end{proof}
\begin{lemma}
The naive smooth histogram approach outputs a solution of size $\mathcal{CP}(BC) = \ell$.
\end{lemma}

\begin{proof}
Recall that all intervals in $B$ and $C$ are inserted only into good regions or at the boundary of good regions.
Let $[x,x+1)$ be a good region (i.e $x = 3k + 2$ for integer $k$).
We will show that $|\mathcal{CP}(BC) \cap [x-1,x+2]| = 3$ (the good region and its neighbouring regions).

After processing the $B$ stream, we have region boundaries at $x + 0.26$,  $ x+0.71$ and $x+1.1$.
Observe that all of the intervals of $C_1$ cross the region boundaries at $x + 0.26$, and $x+0.71$, so they do not get saved by the run of $\mathcal{CP}$.
Furthermore, we have that the interval $[x+0.25,x+0.35] \in C_2$ crosses the region boundary at $x+0.26$ while the interval $[x+0.7,x+0.8] \in C_2$ crosses the boundary at $x+0.71$, so these intervals also do not get saved.

When processing $C$, however, the intervals $[x + 0.06, x + 0.15], [x+0.55,x+0.6], [x+0.9, x+0.94] \in C_2$ get inserted into the solution.
Additionally, they do not change the structure of the regions created by the run $\mathcal{CP}(B)$ (i.e they only modify the leftmost or the rightmost interval of each region created by $B$).

So, $|\mathcal{CP}(BC) \cap [x-1, x+2]| = 3$.
Because there are $\ell/3$ good regions and the intervals $[x-1,x+2]$ where $x = 3k + 2$ do not pairwise intersect, we have $|\mathcal{CP}(BC)| = \ell$ as required.

\end{proof}
\begin{lemma}
Steps $1$ and $2$ of Algorithm \ref{alg:arbitrary-length} output a solution of size $\ell$.
\end{lemma}

\begin{proof}
First, observe that steps 1 and 2 of Algorithm \ref{alg:arbitrary-length} are run on substream $C$. Furthermore, since only good regions contain intervals in substream $C$, it suffices to explore how steps 1 and 2 act on good regions. 

In each good region, the stream $C_1$ is responsible for creating the regions of the runs of Algorithm \ref{alg:arbitrary-length}. Observe that the intervals $[x+0.25, x+0.35] \in C_2$ and  $[x + 0.7, x+0.8] \in C_2$ cross the boundaries of these regions and are thus not stored by the algorithm.
In each good region, only the intervals $[x + 0.06, x + 0.15], [x+0.55,x+0.6], [x+0.9, x+0.94]$ of  $C_2$   get memorized.
Therefore, we obtain a solution of size $3$ for each good region. 
Overall, the obtained solution by the runs of steps 1 and 2 of Algorithm \ref{alg:arbitrary-length} is of size $3 \cdot \frac{\ell}{3} = \ell$.

\end{proof}
Using the last two lemmas, we obtain the following conclusion:

\begin{theorem}
Let $S$ be the size of the solution output by Algorithm \ref{alg:arbitrary-length} on the described input.
Then,$\frac{OPT(ABC)}{S} \ge \frac{11}{3}$.
\end{theorem}

\begin{proof}
By the previous two lemmas, both $\mathcal{CP}(BC)$ and steps 1 and 2 of Algorithm \ref{alg:arbitrary-length} output a solution of size $\ell$.
Notice that the set of saved intervals of steps 1 and 2 of Algorithm \ref{alg:arbitrary-length} is a subset of the saved intervals of $\mathcal{CP}(BC)$, therefore we cannot improve the overall solution by combining both solutions.
So, $S = \ell$.

By Lemma \ref{lemma:optsize}, $OPT(ABC) \ge \frac{11\ell}{3}$. 
So, we get the required conclusion.
\end{proof}

\end{document}